\documentclass[12pt]{article}
\usepackage{xcite}
\usepackage{xr}

\usepackage{graphicx}
\usepackage{natbib} 
\usepackage{url} 
\usepackage{natbib}
\usepackage{amsthm}
\usepackage{url} 
\usepackage[utf8]{inputenc} 
\usepackage[T1]{fontenc}    
\usepackage{hyperref}       
\usepackage{url}            
\usepackage{booktabs}       
\usepackage{amsfonts}       
\usepackage{nicefrac}       
\usepackage{microtype}      
\usepackage{lipsum}
\usepackage{comment}
\usepackage{multirow}
\usepackage{xcolor,ulem}

\newtheorem{lemma}{Lemma}

\usepackage{authblk}
\usepackage{graphicx}
\usepackage{amsmath} 
\usepackage{natbib}
\usepackage{subfigure}
\usepackage{algorithm}
\usepackage{algorithmic}

\newcommand{\bbR}{\mathbb{R}}

\newcommand{\sT}{\mathcal{T}}

\newcommand{\va}{\boldsymbol{a}}
\newcommand{\vb}{\boldsymbol{b}}

\newcommand{\vd}{\boldsymbol{d}}

\newcommand{\vf}{\boldsymbol{f}}

\newcommand{\vh}{\boldsymbol{h}}

\newcommand{\vt}{\boldsymbol{t}}
\newcommand{\vu}{\boldsymbol{u}}
\newcommand{\vv}{\boldsymbol{v}}

\newcommand{\vy}{\boldsymbol{y}}
\newcommand{\vz}{\boldsymbol{z}}

\newcommand{\vtheta}{\boldsymbol{\theta}}
\newcommand{\vbeta}{\boldsymbol{\beta}}
\newcommand{\vpsi}{\boldsymbol{\psi}}

\newcommand{\vzero}{\mathbf{0}}
\newcommand{\vxi}{\boldsymbol{\xi}}

\newcommand{\vtau}{\boldsymbol{\tau}}
\newcommand{\vgamma}{\boldsymbol{\gamma}}

\newcommand{\mA}{\mathbf{A}}
\newcommand{\mB}{\mathbf{B}}
\newcommand{\mC}{\mathbf{C}}
\newcommand{\mD}{\mathbf{D}}
\newcommand{\mE}{\mathbf{E}}
\newcommand{\mF}{\mathbf{F}}

\newcommand{\mH}{\mathbf{H}}
\newcommand{\mI}{\mathbf{I}}

\newcommand{\mK}{\mathbf{K}}

\newcommand{\mS}{\mathbf{S}}
\newcommand{\mT}{\mathbf{T}}

\newcommand{\mV}{\mathbf{V}}
\newcommand{\mW}{\mathbf{W}}

\newcommand{\mTheta}{\boldsymbol{\Theta}}
\newcommand{\mGamma}{\boldsymbol{\Gamma}}

\newcommand{\mSigma}{\boldsymbol{\Sigma}}

\newcommand{\Expect}{\mathbb{E}}
\newcommand{\intd}{\,\mathrm{d}}

\usepackage{color}

\DeclareMathOperator*{\argmin}{arg\,min}
\usepackage{cancel}
\usepackage[toc,title,page]{appendix}

\newcommand{\blind}{0}

\begin{document}

	\def\spacingset#1{\renewcommand{\baselinestretch}%
		{#1}\small\normalsize} \spacingset{1}

	\begin{titlepage}
		\if0\blind
		{
			\title{\bf Functional PCA with Covariate Dependent Mean and Covariance Structure}
			\author{
				\normalsize{Fei Ding\\ 
				Institute of Statistics and Big Data, Renmin University of China, Beijing China, 100872\\
				and \\
				Shiyuan He \\
				Institute of Statistics and Big Data and Center for Applied Statistics, Renmin University of China, Beijing China, 100872\\
				and \\
				David E. Jones \\
				Department of Statistics, Texas A\&M University, College Station, Texas, 77843\\
				and \\
				Jianhua Z. Huang \\
				School of Data Science, The Chinese University of Hong Kong, Shenzhen China, 518172} }
			\date{}
			\maketitle
		} \fi
		
		\if1\blind
		{
			\title{\bf Functional PCA with Covariate Dependent Mean and Covariance Structure}
			\author{}
			\date{}
			\maketitle
		} \fi

		\begin{abstract}
			Incorporating covariates into functional principal component analysis (PCA) can substantially improve the representation efficiency of the principal components and predictive performance. However, many existing functional PCA methods do not make use of covariates, and those that do often have high computational cost or make overly simplistic assumptions that are violated in practice. In this article, we propose a new framework, called Covariate Dependent Functional Principal Component Analysis (CD-FPCA), in which both the mean and covariance structure depend on covariates. We propose a corresponding estimation algorithm, which makes use of spline basis representations and roughness penalties, and is substantially more computationally efficient than competing approaches of adequate estimation and prediction accuracy. A key aspect of our work is our novel approach for modeling the covariance function and ensuring that it is symmetric positive semi-definite. We demonstrate the advantages of our methodology through a simulation study and an astronomical data analysis.
		\end{abstract}

		
	\end{titlepage}
	
	\clearpage
	
	\spacingset{1.5} 
	\section{INTRODUCTION} \label{sec:intro}
	
	Functional data analysis (FDA) is increasingly important in many scientific fields including astronomy, biology, and neuroscience. It is a powerful tool that can be used to jointly model collections of curves, time series, spatial structures, or other functional observations, and can address difficulties such as sparse and irregularly spaced measurements. The key to FDA is exploiting the widespread presence of underlying smoothness in real data to efficiently model similarities and differences between functional observations, e.g., the similarities and differences between times series capturing the changing brightness of stars of a given type.

	Due to the prevalence and variety of functional data, and the computational challenges which arise through their analysis, many FDA methods have been developed. 
	Among these approaches,  the most fundamental is functional principal component analysis (FPCA).
	A comprehensive introduction of FPCA is given in \cite{ramsay2005principal}. In its canonical form, FPCA represents each functional observation as a mean function plus a linear combination of functional principal components (FPCs) and noise. The key to this method is that it is often possible to well capture the data using 
	only a small number of FPCs, i.e., by imposing a low rank structure on the underlying covariance function. Furthermore, by constraining the inferred FPCs to be smooth, we can ensure that the corresponding covariance function is also smooth. Empirical studies have shown that this type of smoothness often improves estimation accuracy and predictive performance due to the usual  bias-variance trade-off encountered in statistical inference.

	One leading FPCA framework is based on the work of \cite{james2000principal}, which interprets FPCA as a mixed effects model for sparse and irregularly sampled  functional data.  
	Their estimation strategy  relies on a basis representation of the underlying eigenfuctions or FPCs, and imposes smoothness  by limiting the number of basis functions in the expansion.   Another popular FPCA approach is to compute the eigenvectors of the sample covariance matrix resulting from a locally smoothed approximation to the underlying covariance function, e.g., see \cite{rice1991estimating,yao2005functional}. A third strategy, proposed by \cite{cai2010nonparametric},  
	relies on a covariance function expansion induced by a user-specified tensor product reproducing kernel Hilbert space (RKHS). This approach has some appealing features, but suffers from the fact that the number of parameters to be optimized grows linearly with the sample size (i.e., the number of curves) and quadratically with the number of observations per curve.  The existing literature has also considered FPCA in the Bayesian setting \citep[e.g.][]{suarez2017bayesian,van2008variational,behseta2005hierarchical}. 
	

	This paper proposes an FPCA method that incorporates covariates in a computationally efficient manner, which is a key extension of the above approaches. Indeed, covariates are often available in practice, and their inclusion can facilitate low-rank representations of functional data and  substantially improve predictive performance.
	There are a number of existing strategies for incorporating covariates, including early methods such as \citet{capra1997accelerated}, and more recent approaches which  primarily build on the work of \cite{yao2005functional}, e.g., \cite{jiang2010covariate,jiang2011functional,zhang2013time,zhang2016sparse}. However, these methods all rely on local smoothing, which is computationally  costly. Thus, in practice, they are slow and  cannot be used for analyzing datasets which are large or have more than a few covariates. The literature also proposes a number of alternative strategies which  achieve greater computational efficiency at the expense of simplifying assumptions. \cite{li2016supervised} assumes that only the FPC scores vary with the covariates linearly, and that the sampling points are the same for each functional observation (often called balanced sampling). 
	However, such assumptions are often violated in practice.

	The approach we propose here overcomes the computational challenges of including covariates, while simultaneously allowing both the mean function and the covariance function to depend on the covariates in a non-linear way, as well as permitting unbalanced sampling patterns. We call our method \textit{Covariate Dependent Functional PCA}, or CD-FPCA for short. Its computational efficiency results from a basis representation of the mean and covariance functions, similar to that used by \citet{james2000principal}, and a modeling approach which carefully avoids costly optimization problems. 
	Non-linear  dependence of the covariance function on the covariates is captured by allowing the underlying FPCs, as well as their scores, to depend flexibly on the covariates.
	Our CD-FPCA approach  contains as a special case the Supervised Sparse and Functional PCA (SupSFPC) method proposed by \citet{li2016supervised}, but avoids the simplifying assumptions mentioned earlier. 
	As a result, our method outperforms SupSFPC in all our numerical studies, as well as  \citet{james2000principal} (no covariate dependence) and \citet{jiang2010covariate}, the latter being  representative of the linear smoothing approaches. 

	A challenge for our approach, and other FPCA methods, is ensuring that the underlying covariance function and FPCs have the required properties, e.g., positive semi-definiteness and orthogonality, respectively. In the \cite{james2000principal} framework, a low-rank covariance function is assumed, meaning that only a small number of FPCs are needed to represent it via the Karhunen-Lo\`{e}ve expansion. This reduces the positive semi-definiteness constraint to ensuring that a small number of eigenvalues are positive. 
	\citet{peng2009geometric}
	addressed the orthogonality constraint by representing the FPCs with a finite basis expansion and 
	using a restricted maximum likelihood (REML) method to fit the basis coefficients. 
	Our approach also relies on a finite basis expansion, but considers the covariance function directly rather than the FPCs. In particular, we use the fact that a low-rank covariance function can be represented by  $G(t,s|\vz)\approx\vb(t)^T \mSigma(\vz)\vb(s)$, where $\vz$ is a covariate vector and  $\vb(t)$ is a vector of the orthonormal basis functions. Then we propose a model for $\mSigma(\cdot)$ that maps  a Euclidean space vector to a symmetric positive semi-definite rank $r$ matrix, thereby automatically ensuring that $\mSigma$ and $G$ are positive semi-definite, and facilitating straightforward model fitting. Some recent papers \citep[e.g.][]{lin2017extrinsic} have explored   the related task of modeling manifold-valued data.
	
	We additionally encourage smoothness of the estimated  functions via 
	the computationally efficient approach of \cite{wood2006low} and \cite{reiss2014varying} for penalizing functions with  multiple arguments.
	This is preferable to controlling smoothness by choosing a small fixed  number of basis functions, because the latter  approach is  discontinuous in nature.  Penalized splines have been used often  in the literature.  With a spline representation and roughness penalty, \cite{reiss2014varying} developed a series of techniques for estimating the mean response, while \cite{greven2017general} proposed a general framework for functional regression. In contrast to those works, our focus  is on modeling the conditional covariance function (i.e., conditional on the covariate value $\vz$) after subtracting the mean function. A key contribution of our work is constructing a proper mapping from  covariates to the conditional covariance function. Given a spline expansion, the problem is equivalent to finding a map from the covariate to a positive semi-definite matrix. We also develop a roughness penalty for this mapping.
	
	

	This article is organized as follows. Section~\ref{sec:model} briefly reviews classical FPCA,  presents our covariate dependent model for the mean and covariance functions, and introduces the roughness penalty. A brief summary of the competing method of \citet{li2016supervised} is also included at the end of Section~\ref{sec:model}. Section~\ref{sec:algorithm} details our algorithm, which makes use of several techniques to reduce computational cost. In Sections~\ref{sec:simulation} and~\ref{sec:application}, we compare our method with the approaches proposed by  \citet{li2016supervised}, \citet{james2000principal},  and \citet{jiang2010covariate} through a simulation study and an astronomical data analysis. Brief discussion is found in Section \ref{sec:discussion}. The Supplementary Materials provide additional numerical results for the simulation study and real data analysis. It also includes the  details of our roughness penalty, technical proofs, and more details of the method of \citet{li2016supervised}.

	\section{COVARIATE DEPENDENT FPCA MODEL}\label{sec:model}
	
	\subsection{Classical FPCA}\label{sec:classical}
	
	Suppose that there are $N$ latent functions of interest, denoted by $x_n$, for $n=1,\dots,N$. Let  $x_n(t)$ denote the value of the $n{\text{-th}}$  latent function at time $t$, for $n = 1, \dots, N$. More generally, we could consider functions of other types of variable, such as location, but here restrict our attention to functions of time. The covariance function $\text{cov}(x_n(t),x_n(s))=G(t,s)$ has eigen-decomposition
	$ G(t,s) = \sum_{j=1}^{\infty}d_jf_j(t)f_j(s)$,
	where $f_j$ is the $j$-th eigenfunction or FPC and $d_j$ is the corresponding eigenvalue. The eigenfunctions are orthonormal to each other, and 
	the eigenvalues are ordered non-increasingly  $d_1 \geq d_2\geq d_3\geq \cdots$.
	By the Karhunen-Lo\`{e}ve theorem, each function $x_n(t)$ can be expressed by a linear combination of a mean function  $\mu(t)$  and the above eigenfunctions, i.e., 
	\begin{align}
		x_n(t) = \mu(t) + \sum_{j=1}^{\infty}\xi_{j}^{(n)}f_j(t),\label{eqn:fpca}
	\end{align}
	where the $\xi_{j}^{(n)} = \int x_n(t)f_j(t) \intd t$ are uncorrelated random variables with mean $0$ and variance $d_j$, for $j=1,2,\dots$. It is often assumed that $\xi_{j}^{(n)}$ follows the Gaussian distribution $\mathcal{N}(0,d_j)$.  The low-rank approximation used by \cite{james2000principal} and others truncates the summation in (\ref{eqn:fpca}), and hence also the eigen-decomposition for the covariance function $G$. This low rank approximation usually works well in practice because  functional data is characterized by its low intrinsic dimension in an infinite dimensional space. Besides,
	we typically only  observe a noisy version of $x_n(t)$, which we denote by $y_n(t)$, for $n=1,\dots,N$. Putting these two points together, we write
	\begin{equation} 
		y_n(t) =x_n(t)+\epsilon_n(t)
		\approx\mu(t)+\sum_{j=1}^r \xi_{j}^{(n)}f_j(t) + \epsilon_n(t)
		=\mu(t) + \vf ^{T}(t) \vxi^{(n)} +\epsilon_n(t),
		\label{EC0}
	\end{equation}
	where $\epsilon_n(t)$ denotes white noise with mean $0$ and variance $\sigma_e^2$. The right-hand side of (\ref{EC0}) uses the vector notation $\vf(t) = (f_1(t),\dots,f_r(t))^{T}$ and $\vxi^{(n)} = (\xi_1^{(n)},\dots,\xi_r^{(n)})^{T}.$  The mixed effects interpretation introduced by \cite{james2000principal} is developed from~(\ref{EC0})  by replacing $\mu(t)$ and $f_j(t)$, for $j=1,\dots,r$, by basis expansions. The resulting expression then has some {\it fixed} basis coefficients and some {\it random} coefficients, hence the mixed effects interpretation. 
	
	
	
	\subsection{Model Extension to Include Covariates}
	
	
	Let $\vz$ denote a vector of covariates, and suppose that $t\in \sT$ and $\vz \in \mathcal{Z}$, where $\sT=[t_{\min},\, t_{\max}]$ and  $\mathcal{Z}$ are compact domains.
	To incorporate covariate dependence, we replace \eqref{EC0} by
	\begin{equation} \label{eqn:supervisedKLoeve}
		y_n(t, \vz) 
		\approx \mu(t,\vz) + \sum_{j=1}^r f_j(t,\vz) \xi_j^{(n)} +\epsilon_n(t).
	\end{equation}
	In this model, both the mean function $ \mu(t,\vz) $ and the FPCs $f_j(t,\vz)$ are allowed to vary smoothly with the covariates $\vz$, as will be explained further in Sections \ref{sec:mean} and \ref{sec:covariance} below. In the rest of this work, we primarily focus on the case where $\vz\in \mathbb{R}$ is a scalar in order to avoid the curse of dimensionality. Extension to the multivariate case is discussed in Section~\ref{sec:discussion}. 
	
	The score vector $\vxi^{(n)}$   follows a multivariate Gaussian distribution with zero mean  and diagonal covariance matrix $\mD_{\vz}$, which also depends on $\vz$, i.e., $\vxi^{(n)}\sim \mathcal{N}_r(\vzero, \mD_{\vz})$, where $\mathcal{N}_r$ denotes a multivariate Gaussian distribution of dimension $r$. For a given covariate  $\vz$,   the rank $r$ approximation to the true covariance function is
	$G(t,s | \vz) \approx \vf ^{T}(t,\vz) \mD_{\vz}\vf (s,\vz)$, with 
	the elements of the vector $\vf (t,\vz)$ being  $f_j(t,\vz)$, for $j=1,\dots,r$.  Thus, the rank of the underlying covariance function is $r$, regardless of $\vz$.
	
	\subsection{Covariate Dependent Mean Function}\label{sec:mean}
	
	In model~\eqref{eqn:supervisedKLoeve}, we  allow the mean function $\mu(t, \vz)$
	to vary smoothly with both time $t$ and the covariates $\vz$, which we achieve using a tensor product spline basis. Firstly, dependence on $t$ is captured by an orthonormal cubic spline with  $l'$ equally spaced knots in the temporal domain $\sT$, i.e., with $l = l'+2$ degrees of freedom. Similarly, dependence on the covariates $\vz$ is captured  using an orthonormal cubic spline with $p'$ knots, and therefore $p = p' + 2$ degrees of freedom.  In our implementation, we employ orthonormalized B-spline basis because this facilitates computing its derivatives \citep{butterfield1976computation} and expressing the roughness penalties later.

	Let  $\va(t)\in \mathbb{R}^{l}$ and $\vu({\vz})\in \mathbb{R}^{p}$ be the values of the  B-spline basis functions evaluated at $t$ and $\vz$, respectively. They are  orthonormalized such that $\int_\sT \va(t)\va(t)^T\mathrm{d}t = \mI_l$ and $\int_\sT \vu(\vz)\vu(\vz)^T\mathrm{d}\vz = \mI_p$, where $\mI_A$ denotes the identity matrix of size $A\times A$, for $A=l,p$. With this notation, our proposed mean function is given by
	\begin{equation}
		\mu(t,\vz) \approx \va(t)^T\mTheta_{\mu} \vu(\vz) =  \sum_{i=1}^l\sum_{j=1}^pa_i(t)u_j(\vz)\theta_{ij} = \mH(t,\vz)^T\vtheta_{\mu},
		\label{meanProduct}
	\end{equation}
	where  $\mH(t,\vz) = \vu(\vz) \otimes \va(t)$ is a tensor product spline basis,  $\mTheta_{\mu} = \big(\theta_{ij} \big)\in \mathbb{R}^{l\times p}$ is the matrix of  basis coefficients, and $\vtheta_{\mu} = \mathrm{vec}(\mTheta_{\mu})$ is its vectorization.

	\subsection{Covariate Dependent  Covariance Function}\label{sec:covariance}

	To model the covariance $G(t,s| \vz)$, we  rely on an orthonormal cubic B-spline basis $\vb(t)\in\bbR^w$ with $w-2$ knots in the temporal domain $\sT$.   The orthonormality of the basis  also helps to easily obtain estimates of the underlying FPCs after training our model. In general we set $w > r$ to facilitate accurate spline approximations of the true eigenfunctions under the assumption of low rank structure.
	Our model covariance function (an approximation to the true covariance $ G(t,s | \vz) $) can now be introduced as
	\begin{equation} \label{Approx2}
		G(t,s | \vz)  \approx \vb(t)^T \mSigma(\vz; \vbeta)\vb(s),
	\end{equation}
	where $\mSigma(\vz; \vbeta)\in\mathbb{R}^{w\times w}$
	is a  matrix that depends on the covariates $\vz$ and is parameterized by coefficients $\vbeta$ to be specified below.

	We must ensure that our covariance function in~\eqref{Approx2} is symmetric positive semi-definite.  This is equivalent to requiring $\mSigma(\vz\, ;\vbeta )$ to be a symmetric positive semi-definite matrix
	for each $\vz$.  Based on the ideas in  \cite{zhu2009intrinsic}, we construct
	$\mSigma(\cdot\, ;\vbeta )$ using a map from  $\mathcal{Z}$ to  symmetric positive semi-definite rank $r$ matrices, i.e.,
	\begin{equation} 
		\mSigma(\cdot;\vbeta):\ \vz \mapsto\mSigma(\vz;\vbeta) = \mC(\vz;\vbeta)\mC(\vz;\vbeta)^T.
		\label{EC1}
	\end{equation}
	In the above, $\mC(\vz; \vbeta) \in \mathbb{R}^{w\times r}$  depends on the covariates $\vz$ and the unknown coefficients  $\vbeta$. 
	The structure~\eqref{EC1} is similar to that of Cholesky decomposition, except that the matrix $\mC(\vz;\vbeta)$ is not required to be lower triangular.

	With the help of~\eqref{EC1}, the  construction of a positive semi-definite 
	$\mSigma(\vz;\vbeta)$
	is reduced to the construction of a general matrix $\mC(\vz;\vbeta)$ \textit{without} restriction. We set its $(i,j)$-th element $C_{ij}(\vz) = \vv(\vz)^T\vbeta_{ij}$, where $\vv(\vz)\in \mathbb{R}^q$  is an orthonormal B-spline basis and  $\vbeta_{ij}\in \mathbb{R}^q$  is the corresponding coefficient, for $i=1,\dots,w$ and $j=1,\dots,r$. In summary,  our model  for  $\mC(\vz;\vbeta)$ (and hence $\mSigma(\vz;\vbeta)$)
	has $w\times r \times q$ unknown parameters, which are  collected in a matrix $\mGamma \in \mathbb{R}^{(qw) \times r}$:
	\begin{equation}
		\mGamma=\left( \begin{matrix}
			\vbeta_{11}&\vbeta_{12}&\cdots&\vbeta_{1r}\\
			\vbeta_{21}&\vbeta_{22}&\cdots&\vbeta_{2r}\\
			\vdots&\vdots&\ddots&\vdots\\
			\vbeta_{w1}&\vbeta_{w2}&\cdots&\vbeta_{wr}
		\end{matrix}\right).
		\label{mGamma}
	\end{equation}
	Let $\vbeta = \text{vec}(\mGamma)$ denote the vectorized version of this matrix.
	It  readily follows that the factor matrix 
	$\mC(\vz; \vbeta)$ in~\eqref{EC1} can be written as $\mC(\vz; \vbeta) = (\mI_w\otimes\vv(\vz)^T)\mGamma$.

	\subsection{Model Negative Log-likelihood} \label{sec:negloglik}

	Based on the above parameterization, our model for $y_n$ in~\eqref{eqn:supervisedKLoeve} an now be written as
	\begin{align}
		y_n(t,\vz) &\approx \mH(t,\vz)^T\vtheta_{\mu} +
		\vb(t)^T\mC(\vz; \vbeta)\vpsi^{(n)}   + \epsilon_n(t),\label{E3}
	\end{align}
	where $\vpsi^{(n)}\sim\mathcal{N}_r(\textbf{0}, \mI_r)$, with $\textbf{0}$ being the zero vector of length $r$, and $\epsilon_n(t) \sim \mathcal{N}(0,\sigma_e^2)$ is white noise.  It follows that $y_n(t,\vz)$ with fixed $\vz$ is a Gaussion process with mean $\mH(t,\vz)^T\vtheta_{\mu}$ and covariance function $\vb(t)^T\mSigma(\vz;\vbeta) \vb(s) + \sigma_e^2I(s=t)$.  Given the matrices and parameters in~\eqref{E3}, the  original eigenfunctions $\vf (s,\vz)$ and scores $\vxi^{(n)}$ in~\eqref{eqn:supervisedKLoeve} can be easily recovered at a fixed  $\vz$ via the next lemma. See Section~S.7.1 of the Supplementary Materials for its proof and further discussion.
	\begin{lemma} \label{lemma:eigenrelation}
		Let $\mSigma(\vz;\vbeta) =  \mTheta_{\vz} \mD_{\vz}\mTheta_{\vz}^T$ be the  eigen-decomposition 
		of the matrix in~\eqref{EC1}, where $\mD_{\vz}\in\mathbb{R}^{r\times r}$ is a diagonal matrix with the non-zero eigenvalues and $\mTheta\in\bbR^{w\times r}$ contains the corresponding eigenvectors in its columns. Denote $\mTheta_{\vz,j}$  as the $j$th eigenvector, then $f_j(t,\vz) = \vb^T(t) \mTheta_{\vz,j}$ for $j=1,\dots, r$. Besides, the  correspondence relation between $\vxi^{(n)}$ and $\vpsi^{(n)}$ is given by $\vxi^{(n)}=\mD_{\vz}^{1/2}\mV_{\vz}^T\vpsi^{(n)}$ where
		$\mV_{\vz} = [\mC(\vz;\vbeta)^T\mC(\vz;\vbeta)]^{-1}\mC(\vz;\vbeta)^T\mTheta_{\vz}\mD_{\vz}^{1/2}$. 
	\end{lemma}

	In practice, $y_n$ is only observed at a finite collection of observation times $\vt_n=(t_1^{(n)},\dots,t_{m_n}^{(n)})$, and has a specific value of the covariates associated with it, which we denote by $\vz_n$. 
	To simplify notation we collect the  basis evaluations $\vb(t_i^{(n)})$ for the covariance function (see (\ref{Approx2})) and the tensor product basis evaluations $\mH(t^{(n)}_i,\vz_n)$ for the mean function (see (\ref{meanProduct})) into matrices $\mB_n$ and $\mH_n$, respectively, i.e.,
	\begin{equation} \label{eqn:designMatB}
		\mB_n= (\vb(t^{(n)}_1),\vb(t^{(n)}_2),\dots,\vb(t^{(n)}_{m_n}))^T,
	\end{equation}
	and
	\begin{equation} \label{eqn:designMatH}
		\mH_n = (\mH(t^{(n)}_1,\vz_n),\mH(t^{(n)}_2,\vz_n),\dots,\mH(t^{(n)}_{m_n},\vz_n))^T.
	\end{equation}
	Thus, the observations $\vy_n = (y_n(t_1^{(n)}),\dots, y_n(t_{m_n}^{(n)}))^T$ follow a multivariate Gaussian with mean 
	$\mH_n \vtheta_{\mu}$ and covariance matrix $\mSigma_n = \mB_n \mC_n\mC_n^T\mB_n^T +  \sigma_e^2 \mI_{m_n}$, where $\mC_n = \mC(\vz_n,\vbeta)$.
	With this notation, the  negative log-likelihood of the full dataset  is   proportional to
	\begin{equation}
		\mathcal{L}(\vtheta_{\mu}, \vbeta, \sigma_e^2 ) := \sum_{n=1}^N  \log \det \mSigma_n + \mathrm{tr}(\mS_n\mSigma_n^{-1}),
		\label{likelihood}
	\end{equation}
	where $\mS_n = (\vy_n - \mH_n\vtheta_{\mu})(\vy_n- \mH_n\vtheta_{\mu})^T$. 
	

	To encourage smoothness of the estimated mean  and covariance functions, we follow the roughness penalty approach of \cite{wood2006low} and \cite{reiss2014varying}. The penalty construction are given in Section~S.3 of the Supplementary Materials. Combining the negative log-likelihood (\ref{likelihood}) with the roughness penalties given by~(S.10) and~(S.11), we obtain the following objective function to be minimized with respect to the parameters $\vtheta_{\mu}, \vbeta, \sigma_e^2$:
	\begin{align}
		\mathcal{L} + \mathcal{P} &=\sum_{n=1}^N  \{\log \det \mSigma_n + \mathrm{tr}(\mS_n\mSigma_n^{-1})\}\nonumber\\
		&	\qquad\qquad\qquad+\vtheta^T_{\mu}(\lambda^{(\mu)}_t \widetilde{\mS}_t^{(\mu)}+ 	\lambda^{(\mu)}_{\vz}\widetilde{\mS}_{\vz}^{(\mu)})\vtheta_{\mu}+
		\vbeta^T(\lambda_t\mI_r\otimes\widetilde{\mS}_t+
		\lambda_{\vz}\mI_r\otimes\widetilde{\mS}_{\vz})\vbeta,
		\label{objective}
	\end{align}
	where $\mathcal{L}$ and $\mathcal{P}$ denote the log-likelihood and penalty terms, respectively. 

	According to the penalized spline literature, the number of spline basis is not important and does not need tuning as long as a relatively large number of spline basis functions is used, and the penalty parameter does the fine tuning of amount of smoothing. We thus  fix the number of spline basis functions to be a relatively large number. It only remains to choose the four tuning parameters ($\lambda$'s)  and the number $r$ of FPCs. Selecting among all  possible candidate combinations of the tuning parameters is time consuming. We therefore adopt the following strategy. We first select the tuning parameters ($\lambda_t^{(\mu)}$ and $\lambda_{\vz}^{(\mu)}$) for  the mean function via  generalized cross-validation. After fixing $\lambda_t^{(\mu)}$ and $\lambda_{\vz}^{(\mu)}$, we select the tuning parameters ($\lambda_t$ and $\lambda_{\vz}$)  for the FPCs  via $K$-fold cross-validation with the proposed loss~\eqref{likelihood}. Finally, the number of FPCs can be chosen based on  a scree-plot or the fraction of  variance explained (FVE). More details can be found in Section~S.4 of the Supplementary Materials.

	\subsection{Connection with SupSFPC}
	
	We conclude this section by clarifying the distinction between our approach and Supervised Sparse and Functional PCA  \citep[SupSFPC,][]{li2016supervised}, because the latter  is the only alternative scalable method that incorporates covariate information.  The   SupSFPC  method assumes that the score vector $\vxi^{(n)}$ in~\eqref{EC0} is linearly related to the covariates $\vz_n$, i.e.,
	\begin{align}   \label{eqn:SupSFPC:score}
		\vxi^{(n)} = \vtau_0 + \mT^T\vz_n + \vgamma^{(n)},
	\end{align}
	for  an  intercept vector $\vtau_0\in \mathbb{R}^r$ and  a coefficient matrix $\mT$. The vector $\vgamma^{(n)}\in \mathbb{R}^r$ follows a Gaussian distribution  with mean zero and diagonal covariance matrix $\mSigma_{\vgamma}$. Under this assumption, \eqref{EC0} can be written as
	\begin{align}  \label{eqn:SupSFPC:model}
		y_n(t)  =  \mu(t) + \vtau_0^T\vf(t)+ \vz_n^T\mT\vf(t) + \big[(\vgamma^{(n)})^T\vf(t) + \epsilon_n(t)\big].
	\end{align}
	This implies the conditional distribution $y(t) | \vz_n$ follows a Gaussian process with mean 
	\begin{align}  \label{eqn:SupSFPC:mean}
		\Expect (y_n(t)|\vz_n) = \mu(t) + \vtau_0^T\vf(t)+ \vz_n^T\mT\vf(t),
	\end{align} 
	and covariance function $\text{cov}(y_n(t),y_n(s)|\vz_n) =\vf(t)^T\mSigma_{\vgamma}\vf(s)+\sigma_e^2 I(s=t)$. Note that under the SupSFPC model the covariance function as a constant function with respect to the $\vz$. Ignoring the sparse structure
	and its extra requirement of common measurement time points across all curves,
	the SupSFPC model can be considered  a special case of our model. Indeed, the spline basis used by CD-FPCA can incorporate a linear relation between the mean function and the covariate $\vz_n$, as well as a constant relation between the covariane function and the covariate $\vz_n$. 
	Our full model is well suited to modeling both a smooth covariate-driven mean function {\it and} a smooth covariate-driven covariance function.   More details of  SupSFPC are provided in Section~S.5  of the Supplementary Materials.

	\section{ALGORITHM}\label{sec:algorithm}
	
	
	\subsection{Model Training}\label{subsectionModeltraining}

	Given noisy observations of a collection of latent functions, we can estimate the mean and covariance  functions by optimizing~\eqref{objective}.  
	However, optimization is challenging  because the objective function~\eqref{objective} is non-convex. Moreover, evaluating the log-likelihood $\mathcal{L}$ and its gradients  involves computing the inverses of
	$\mSigma_n\in \bbR^{m_n\times m_n}$, for $n=1,\dots,N$, which has a combined computational cost of $\mathcal{O}(\sum_n m_n^3)$. In practice, this latter issue is exacerbated by the need to perform exploratory evaluations of the objective function to select a good  step size, i.e., the tuning parameter that controls the magnitude of changes in the parameters in each iteration of the optimization algorithm. We overcome these problems by proposing an
	initialization algorithm which  identifies good 
	initial parameter values, and  by developing efficient ways to evaluate the objective function \eqref{objective} and its gradient. 
	Our gradient descent based optimization strategy is summarized in Algorithm~\ref{codeEntireAlgorithm}. The algorithm iteratively updates the parameters $\vtheta_{\mu}$, $\vbeta$, and $\sigma_e^2$  until convergence.

	
	In what follows, we repeatedly apply 
	the matrix determinant lemma and Sherman-Morrison-Woodbury formula to reduce the per curve cost of computing the matrix inverse and log-determinant in Steps 4-6 of Algorithm 1 from  $\mathcal{O}(m_n^3)$ to $\mathcal{O}(r^3)$.
	Initialization of  $\vbeta$ (Step 1) is discussed at the end of this subsection.
	Proofs for all the results below are given in Section~S.7 of the Supplementary Materials. We begin with Lemma~\ref{lemmloglikehood} which presents  a more efficient expression for the log-likelihood $\mathcal{L}$ appearing in \eqref{objective} (also see \eqref{likelihood}).

	\begin{algorithm}[t]
		\caption{Modified gradient descent for optimizing CD-FPCA objective  (\ref{objective})}
		\begin{algorithmic}[1]
			\STATE Initialize $\vbeta$ using Algorithm \ref{code:Initial} (below);
			\STATE Initialize  $\vtheta_{\mu}$  to be the zero vector, and $\sigma^{2}_e$ to be an appropriate small positive value;
			\REPEAT
			\STATE  Update $\vtheta_{\mu}$ by  gradient descent until convergence, gradient is sum of \eqref{fgratheta} and~\eqref{eqn:penaltyGrad:theta};
			\STATE  Update $\vbeta$ by  gradient descent until convergence, gradient is sum of ~\eqref{grabeta} and~\eqref{eqn:penaltyGrad:beta};
			\STATE  Update $\sigma^{2}_{e}$  by  gradient descent until convergence, gradient is given by (\ref{fgrasigma});
			\UNTIL{convergence}
		\end{algorithmic}
		\label{codeEntireAlgorithm}
	\end{algorithm}

	\begin{lemma}\label{lemmloglikehood}
		The loss function $\mathcal{L}$ given by (\ref{likelihood}) is equal to
		\begin{align}
			2\sum_{n=1}^N \log\det(\mF_n)
			-\sigma_e^{-4} \sum_{n=1}^N \Vert \vh_n\Vert_2^2
			+ \sigma_e^{-2} \sum_{n=1}^N \Vert \vy_n-\mH_n\vtheta_{\mu}\Vert_2^2+\sum_{n=1}^Nm_n\log\sigma_e^2,
		\end{align}
		where $\mF_n$ is the Cholesky factor of $\mI_r+\sigma_e^{-2}\mW_n$ (i.e., $\mF_n \mF_n^T = \mI_r+\sigma_e^{-2}\mW_n$),    with $\mW_n = \mC_n^T\mB_n^T\mB_n \mC_n$, and $\vh_n = \mF_n^{-1}\mC^T_n \mB_n^T(\vy_n-\mH_n\vtheta_{\mu})$.
	\end{lemma}

	Next, it is straightforward to verify that the gradients of the log-likelihood (\ref{likelihood})  with respect to $\vtheta_{\mu}$ and $\sigma_e^2$ are
	\begin{equation}
		\frac{\partial{\mathcal{L}}}{\partial{\vtheta}_{\mu}} = \sum_{n = 1}^{N}2\mH_n^T\mSigma_n^{-1}(\mH_n\vtheta_{\mu}-\vy_n),
		\label{gratheta}
	\end{equation}
	and
	\begin{equation}
		\frac{\partial \mathcal{L}}{\partial \sigma_e^2} =
		\sum_{n=1}^{N} \mathrm{tr}(\mSigma_n^{-1})-\sum_{n=1}^N(\vy_n-\mH_n\vtheta_{\mu})^T\mSigma_n^{-2}(\vy_n-\mH_n\vtheta_{\mu}),
		\label{grasigma}
	\end{equation}
	respectively. Following a  similar approach as for Lemma~\ref{lemmloglikehood},  these gradients can be expressed in a computationally more efficient way.

	\begin{lemma}\label{lemmagratheta}
		Let	 $\mF_n$ and $\vh_n$ be defined as in Lemma~\ref{lemmloglikehood} above.	The gradient $\frac{\partial \mathcal{L}}{\partial \vtheta_{\mu}}$ given by (\ref{gratheta}) can be expressed as
		\begin{align}
			\frac{\partial \mathcal{L}}{\partial \vtheta_{\mu}}= \sum_{n=1}^N-2\sigma_e^{-2}\mH_n^T (\vy_n-\mH_n\boldsymbol{\vtheta}_{\mu})+2\sigma_e^{-4}\mH_n^T \mE_n^T\vh_n,
			\label{fgratheta}
		\end{align}
		where $\mE_n = \mF_n^{-1}\mC_n^T\mB_n^T$.
		
	\end{lemma}
	
	\begin{lemma}\label{lemmagrasigma}
		The gradient $\frac{\partial \mathcal{L}}{\partial \sigma_e^2}$ given by (\ref{grasigma}) can be expressed as
		\begin{align}
			\frac{\partial \mathcal{L}}{\partial \sigma_e^2}&=	 \sum_{n=1}^N\Big[ m_n\sigma_e^{-2}-\sigma_e^{-4}\mathrm{tr}(\mE_n^T\mE_n)  - \sigma_e^{-4}\Vert \vy_n-\mH_n\vtheta_{\mu}\Vert_2^2 \nonumber \\
			&\qquad  \qquad\qquad\qquad \qquad\qquad
			+2 \sigma_e^{-6}\|\vh_n\|_2^2-\sigma_e^{-8}\vh_n^T\mE_n\mE_n^T\vh_n\Big],
			\label{fgrasigma}
		\end{align}
		where $\vh_n$ and $\mE_n$ are as  given in Lemma \ref{lemmloglikehood} and Lemma \ref{lemmagratheta} above.
	\end{lemma} 
	
	We now consider the gradient of $\mathcal{L}$ with respect to $\vbeta$.
	Let $\beta_{ijk}$ denote the
	$k$-th element of the vector $\boldsymbol{\beta}_{ij}$. We have
	\begin{equation}
		\frac{\partial \mathcal{L}}{\partial \beta_{ijk}} =
		\sum_{n=1}^N \Big\langle \frac
		{\partial{\mathcal{L}}}{\partial{\mC_n}}, \frac{\partial \mC_n}{\partial \beta_{ijk}}\Big\rangle,
		\label{grabeta}
	\end{equation}
	where the inner product is defined as $\left\langle \mA, \mB \right\rangle = \mathrm{tr}(\mA^T\mB)$, $\frac{\partial{\mC_n}}{\partial{\beta_{ijk}}}$ is the matrix of zeros except that its $(i,j)$ element is $v_k(\vz)$, which is the $k$-th element of the basis $\vv(\vz)$. It also holds that 
	\begin{align}
		\frac{\partial{\mathcal{L}}}{\partial{\mC_n}} = 2\times \mB_n^T\big[\mSigma_n^{-1}-\mSigma_n^{-1}\mS_n\mSigma_n^{-1}\big]\mB_n\mC_n.
		\label{R1}
	\end{align}
	Lemma~\ref{lemmagrabeta} below gives a more computationally efficient expression for~\eqref{R1}.
	
	\begin{lemma}\label{lemmagrabeta}
		The gradient $\frac{\partial{\mathcal{L}}}{\partial{\mC_n}}$ given by (\ref{R1}) can be expressed as
		\begin{align}
			\frac{\partial{\mathcal{L}}}{\partial{\mC_n}} &= 
			2\sigma_e^{-2}(\mB_n^T\mB_n\mC_n-\mB_n^T\mK_n\mW_n)\nonumber\\
			&\qquad\qquad-2\sigma_e^{-4}(\mB_n^T-\mB_n^T\mK_n\mC_n^T\mB_n^T)\mS_n(\mB_n\mC_n-\mK_n\mW_n),
		\end{align}
		where $\mW_n = \mC_n^T\mB_n^T\mB_n \mC_n$ and $\mK_n = \mB_n\mC_n\{\sigma_e^2\mI_r+\mW_n\}^{-1}$.
	\end{lemma}
	
	Lastly, the gradients of the penalty  term $\mathcal{P}$  in~ (\ref{objective}) with respect to   $\vtheta_{\mu}$ and  $\vbeta$ are
	\begin{equation} \label{eqn:penaltyGrad:theta}
		\frac{\partial \mathcal{P}}{\partial \vtheta_{\mu}} = 2\lambda^{(\mu)}_t  \widetilde{\mS}_t^{(\mu)} \vtheta_{\mu} +  2\lambda^{(\mu)}_{\vz} \widetilde{\mS}_{\vz}^{(\mu)} \vtheta_{\mu}
	\end{equation}
	and
	\begin{equation} \label{eqn:penaltyGrad:beta}
		\frac{\partial \mathcal{P}}{\partial \vbeta} =
		2	\lambda_t (\mI_r\otimes\widetilde{\mS}_t)\vbeta+
		2	\lambda_{\vz} (\mI_r\otimes\widetilde{\mS}_{\vz}) \vbeta,
	\end{equation}
	respectively. The gradients of the objective function \eqref{objective} with respect to $\vtheta_{\mu}$ and $\vbeta$ are  obtained by summing~\eqref{gratheta} and~\eqref{eqn:penaltyGrad:theta} and summing
	~\eqref{grabeta} and~\eqref{eqn:penaltyGrad:beta}, respectively.

	%

	It remains to specify a procedure  to initialize  $\vbeta$ in Step 1 of Algorithm~\ref{codeEntireAlgorithm}, and our approach is summarized in Algorithm~\ref{code:Initial}. Step 1 of Algorithm~\ref{code:Initial} divides the covariate domain into small bins (or regions) $\mathcal{Z}_1,\dots,\mathcal{Z}_{U}$, and treats the observations in each bin as having a fixed covariates vector ${\vz}_{u}$, for $u = 1\dots,U$. The  covariate vectors ${\vz}_{u}$, for $u = 1\dots,U$, are set to be the mean observed covariate vector in each bin, i.e., $\vz_u := n_u^{-1}\sum_{n: \vz_n \in \mathcal{Z}_u}\vz_n$, where $n_u= |\mathcal{Z}_u |$ is the cardinality of the set $\mathcal{Z}_u$. In Step 2, for each bin, we  fit the classical FPCA model introduced in \eqref{EC0} to the subset of  functional observations falling in that particular  bin, i.e., we fit it separately to $\{(\vt_n,\vy_n):\vz_n\in\mathcal{Z}_u\}$, for each $u\in\{1,\dots,U\}$. Thus, we obtain  an estimated covariance function $\hat{G}(t,s| \vz_u) = \vb(t)^T \widehat{\mSigma}_{\vz_u} \vb(s)$ for each fixed $\vz_u$. 
	Finally,  making use of \eqref{EC1}, Step 3 initializes  $\vbeta$  by
	\begin{equation} \label{leasquare}
		\vbeta^{(0)} := \argmin_{\vbeta} \sum_{u=1}^U \|\widehat{\mSigma}_{\vz_u}^{1/2} - \mC(\vz_u;\vbeta)\|_F^2,
	\end{equation}
	where $\|\cdot\|_F$ denotes the Frobenius norm. 
	
	\begin{algorithm}[t]
		\caption{Initialization of $\vbeta$}
		\begin{algorithmic}[1]
			\STATE Divide the covariates domain into small bins (or regions) $\mathcal{Z}_1,\dots,\mathcal{Z}_{U}$.
			\STATE Fit the classical FPCA model \eqref{EC0} in each bin to obtain $\widehat{\mSigma}_{\vz_u}$, for $u=1,\dots,U$.
			\STATE Initialize $\vbeta$ by $\vbeta^{(0)}$ in (\ref{leasquare}).
		\end{algorithmic}
		\label{code:Initial}
	\end{algorithm}

	\subsection{Prediction}\label{subsectionScorePrediction}
	
	Suppose that we have applied Algorithm \ref{codeEntireAlgorithm} to a training dataset to get estimation of the related parameters. Now, we obtain noisy observations $\vy_{*} = (y_{*}(t_1^{(*)}),\dots, y_{*}(t_{m_{*}}^{(*)}))^T$ of a new latent function $x_{*}(t)$ at the time points $t_1^{(*)}, \dots,t_{m_{*}}^{(*)}$, together with a corresponding covariate  $\vz_{*}$. We want to estimate the scores $\vxi^{(*)}=(\xi_1^{(*)},\dots,\xi_r^{(*)})^T$ for the function $x_{*}(t)$ (see \eqref{eqn:supervisedKLoeve}), and thereby predict the value of a new  observation $y_{*}(t)$ at any $t\in\mathcal{T}$.
	
	Following \citet{yao2005functional}, our approach to this prediction task is motivated by the  empirical Bayes perspective. Let $\hat{\vtheta}_{\mu}$, $\hat{\vbeta}$ and  $\hat{\sigma}_e^2$ denote the estimates of the model parameters (see \eqref{E3}) obtained by applying Algorithm \ref{codeEntireAlgorithm} to the training dataset. For the new observations, we define basis matrices
	$\mB_*= (\vb(t^{(*)}_1),\dots,\vb(t^{(*)}_{m^*}))^T$
	and
	$\mH_* = (\mH(t^{(*)}_1,\vz_*), \dots,\mH(t^{(*)}_{m^*},\vz_*))^T$, analogous to \eqref{eqn:designMatB} and~\eqref{eqn:designMatH}, respectively. Next, we compute the eigendecomposition 
	$\mSigma(\vz_{*};\hat{\vbeta} ) = \mTheta_{*} \mD_{*}\mTheta_{*}^T$.
	Treating the parameter estimates $\hat{\vtheta}_{\mu}$, $\hat{\vbeta}$ and  $\hat{\sigma}_e^2$ as if they were the  the true values (i.e., the plug-in approach),  the joint distribution of $\vy_{*}$ and $\vxi^{(*)}$ is
	\begin{equation}
		\begin{pmatrix}
			\vy_{*} \\ \vxi^{(*)}
		\end{pmatrix}\sim \mathcal{N}_{2m^*}\left(
		\begin{pmatrix}
			\mH_{*}\hat{\vtheta}_{\mu} \\ \vzero
		\end{pmatrix},
		\begin{pmatrix}
			\mB_{*}
			\mSigma(\vz_{*}; \hat{\vbeta}) \mB_{*}^T+\hat{\sigma}_e^2\mI &&&
			\mB_{*}\mTheta_{*} \mD_{*} \\
			\mD_{*}\mTheta_{*} ^T\mB_{*}^T &&& \mD_{*}
		\end{pmatrix}\right).
		\label{pre_coef_matrix}
	\end{equation}
	This joint distribution results from assuming the prior $\vxi^{(*)}\sim \mathcal{N}( \vzero, \mD_{*})$, which is derived from the training data (and $\vz^*$), hence the empirical Bayes connection.  In a more complete empirical Bayes treatment $\vtheta_{\mu}$, $\vbeta$ and  $\sigma_e^2$ would also be assigned priors, but this introduces additional complications and computation  and is therefore avoided here. Based on \eqref{pre_coef_matrix}, the posterior distribution of $\vxi^{(*)}$ is again a multivariate Gaussian whose mean and covariance matrix are given by
	\begin{equation}
		\mathbb{E}(\vxi^{(*)}|\vy_{*}) = \mD_{*}\mTheta_{*}^T\mB_{*}^T(\mB_{*}
		\mSigma(\vz_{*}; \hat{\vbeta}) \mB_{*}^T+\hat{\sigma}_e^2\mI)^{-1}(\vy_{*}-\mH_{*}\hat{\vtheta}_{\mu})
		\label{E4}
	\end{equation}
	and
	\begin{equation}
		\mathrm{Cov}(\vxi^{(*)}|\vy_{*})  = \mD_{*} - \mD_{*}\mTheta_{*}^T\mB_{*}^T(\mB_{*}   \mSigma(\vz_{*}; \hat{\vbeta}) \mB_{*}^T + \hat{\sigma}_e^2\mI)^{-1}\mB_{*}\mTheta_{*} \mD_{*},
		\label{post_cov}
	\end{equation}
	respectively.

	Finally, combining \eqref{E4} and \eqref{post_cov} with  \eqref{E3}, the posterior predictive distribution of $y_{*}(t)$ at a new time $t$ is a univariate Gaussian distribution with mean and variance given by
	\begin{equation}
		\mH(t,\vz_{*})\hat{\vtheta}_{\mu} + \vb(t)^T\mTheta_{*}\mathbb{E}(\vxi^{(*)}|\vy_{*})
		\label{PreExpect}
	\end{equation}
	and
	\begin{equation}
		\vb(t)^T\mTheta_{*}\mathrm{Cov}(\vxi^{(*)}|\vy_{*})\mTheta_{*}^T\vb(t) + \hat{\sigma}_e^2,
		\label{Prevar}
	\end{equation}
	respectively. 
	If we are instead interested in the underlying latent function value $x_{*}(t)$, then the posterior predictive distribution will be the same except that  the expression for the variance will not have the $\hat{\sigma}_e^2$ term. Sometimes, such as in astronomy, measurement errors are provided with each observed value of $y_{*}(t)$. In this case, we modify our predictions by replacing all instances of $\hat{\sigma}_e^2$ above by
	the actual measurement error value  (including in the application of Algorithm \ref{codeEntireAlgorithm} to the training data). 
	
	\section{SIMULATION STUDY}\label{sec:simulation}
	
	We now compare our CD-FPCA  (Algorithm \ref{codeEntireAlgorithm}) with the methods proposed by  \citet{james2000principal},  \citet{jiang2010covariate}, and \citet{li2016supervised}. The \citet{james2000principal} method uses a spline basis  to approximate the classical FPCA model \eqref{EC0}, but does not incorporate covariates. Following \citet{jiang2010covariate}, we denote this approach by rFPCA, where the "r" stands for reduced rank. \citet{jiang2010covariate} proposed two local linear smoother based  methods, which do incorporate covariate information, but with high computational cost. Their  methods are  called fully adjusted FPCA (fFPCA) and mean adjusted FPCA (mFPCA); the former allows both the mean and covariance function to depend on covariates, whereas the latter only allows the mean function to do so.    The supervised sparse and functional principal component (SupSFPC) method proposed by \cite{li2016supervised} allows  the  scores $\vxi^{(n)}$ in \eqref{EC0} (but not the mean function) to vary with the covariates, and is computationally more efficient than all the other methods considered here (including ours). It is a state-of-the-art approach, and therefore a key comparison. Nonetheless, it has a number of limitations including the linear assumption and a requirement that the data have balanced sampling and regular spacing. 
	
	\subsection{Simulated Datasets}\label{subsec:simulationdesign}
	
	We simulate two datasets of noisy realizations of  $N=100$ and $N = 500$ latent functions, respectively. In both datasets, the $n$-th latent function $x_n(t,z)$ is a linear combination of a mean function $\mu(t,z)$ and $r=3$ orthonormal  eigenfunctions $f_j(t,z)$, $j = 1, \dots, r$, see \eqref{eqn:supervisedKLoeve}.  We set the covariate $z$ to be univariate, the mean function to be
	$\mu(t,z) = 30(t-z)^2$,
	and the three eigenfunctions to be
	$
	f_1(t,z) = \sqrt{2}\cos(\pi(t+z)),\label{eigen1}$ $
	f_2(t,z) = \sqrt{2}\sin(\pi(t+z))\label{eigen2}$ and $
	f_3(t,z) =
	\sqrt{2}\cos(3\pi(t-z)).\label{eigen3}$
	To further impose dependence of the covariance structure on the covariate $z$ we set the eigenvalues to be
	$\vd_z = \big(2(z+20),z+10,z\big)$.
	The scores $\vxi^{(n)}$ are sampled from a Gaussian distribution with mean $\boldsymbol{0}$ and covariance matrix $\mD_{z} = \mathrm{diag}(\vd_z)$.

	For easy comparison, all data generated in this section lie on a regular grid, i.e., the time points $t_i = \frac{i-1}{m-1}$ for $i=1,\dots,m = 100$. This accommodates the SupSFPC method which cannot handle irregularly spaced functional data. The final simulated dataset contains $m$ noisy observational points of each latent function $x_n$, i.e.,
	\begin{align}
		y_n(t_i,z) &= x_n(t_i,z)+\epsilon_n(t_i)=\mu(t_i,z)+\sum_{j=1}^r \xi^{(n)}_{j}f_j(t_i,z) + \epsilon_n(t_i)\label{equationx1} 
	\end{align}
	for $i=1,\dots,m$, where  $\epsilon_n(t_i)$ is an independent white noise  with variance $\sigma_e^2=0.01$. We repeat the simulation $50$ times.

	\subsection{Results}\label{Results_two}

	For our CD-FPCA method, we set the number of B-spline basis functions for capturing the dependence of the mean function on $t$  and $z$  to be  $l =10$ and $p = 10$, see (\ref{meanProduct}). For the covariance function, we set the number of basis functions for capturing dependence on  $t$  and $z$ to be $w=10$ and $q=10$, see Section~\ref{sec:covariance}. Similarly, for rFPCA we  use  $10$ basis functions to capture the dependence of the mean and covariance function on $t$  (the rFPCA model does not incorporate covariates). 
	For the mFPCA and fFPCA methods \citep{jiang2010covariate} we apply $10$-fold cross-validation to the $N=100$ dataset to select the smoothing bandwidths (cross-validation is too time-consuming to apply to the $N = 500$ dataset so we use the same values  for that dataset as well). In particular, to estimate the mean function we use the bandwidths  $0.57$ and  $0.52$ to smooth across time $t$ and the covariate $z$, respectively. To estimate the eigenfunctions we use the bandwidths $0.57$ and $0.66$ to smooth across time $t$ and the covariate $z$ (fFPCA only), respectively.

	\begin{table}[t]
		\footnotesize
		\caption{MSE   of the mean and eigenfunction estimators under CD-FPCA, SupSFPC, fFPCA, mFPCA, and rFPCA.}
		\centering
		\begin{tabular}{c||c|c|c|c|c}
			\toprule
			\cmidrule(r){1-6}
			\multicolumn{2}{c}{}&\multicolumn{4}{c}{MSE (SE)} \\
			\cmidrule(r){3-6}
			\multicolumn{2}{c}{}&Mean Fun.&First Eigen. &Second Eigen.&\multicolumn{1}{c}{Third Eigen.}\\
			
			\hline
			\multirow{5}{*}{$N=100$}   
			&CD-FPCA &$3.31$ ($0.25$)&$0.309$ ($0.036$)& $0.326$ ($0.035$) &$0.036$ ($0.004$) \\
			&SupSFPC  &$6.27$ ($0.15$)&$0.726$ ($0.008$) &$0.738$ ($0.008$) &$0.849$ ($0.013$)\\
			&fFPCA  &$5.05$ ($0.30$) &$0.396$ ($0.024$) &$0.427$ ($0.023$) &$1.855$ ($0.004$) \\
			&mFPCA  &$5.43$($0.30$) &$0.732$($0.009$) &$0.734$($0.009$) &$1.867$($0.003)$ \\
			&rFPCA &$31.24$ ($0.35$) &$0.715$ ($0.008$) &$0.738$ ($ 0.009$) &$0.855$ ($0.012$) \\		
			\hline
			\multirow{5}{*}{$N=500$}   &CD-FPCA  &$1.81$ ($0.15$) & $0.038$ ($0.004$)& $0.040$ ($0.004$)& $0.004$ ($0.000$) \\
			&SupSFPC  &$ 5.19$ ($0.05$) &$0.722$ ($0.004$) &$0.735$ ($0.004$) &$0.861$ ($0.006$)\\
			&fFPCA  &$3.84$ ($0.09$) &$0.332$ ($0.015$) &$0.356$ ($0.015$) &$1.869$ ($0.002$)\\
			&mFPCA  & $4.29$ ($ 0.09$)  & $ 0.731$ ($0.004$)& $0.735$ ($0.004$)& $1.879$ ($0.001$)\\
			&rFPCA & $30.86$ ($0.19$) & $0.710$ ($0.004$)& $0.734$ ($0.004$)& $0.868$ ($0.006$)\\
			\hline    
		\end{tabular}
		\label{tab:real_comparison}
	\end{table}

	The five methods are compared based on  the mean squared errors (MSE)  for their estimated mean function and eigenfunctions. For a function $g$ and its estimate $\hat{g}$, we define the squared error to be $\frac{1}{Nm}\sum_{n=1}^N\sum_{i=1}^{m} (g(t_i,z_n)-\hat{g}(t_i,z_n))^2$, where $m$ is the number of grid points at which each function is observed. We approximate the MSE by the mean of the squared error across the $50$ replicate simulations. For SupSFPC, the MSE for the mean function is computed over the conditional mean given by~\eqref{eqn:SupSFPC:mean}. When computing the loss for the estimated eigenfunctions, the $\pm 1$ sign needs to be matched to the true eigenfunction because eigenfunctions are only identifiable up to a $\pm 1$ sign. 
	
	Table~\ref{tab:real_comparison}  shows the  mean squared errors (MSE) and its standard error (SE) across the $50$ simulations. Our CD-FPCA method has the lowest MSE than all the other approaches for all three eigenfunctions and the mean function. The second best method in terms of MSE is  fFPCA. The remaining methods (mFPCA, SupSFPC and rFPCA) are limited by their underlying modeling assumptions. They do not have the flexibility  to recover the true mean function and the  true eigenfunctions in our simulation setting, which explains why their estimation accuracy does not decrease as the sample size increases. Although fFPCA has the same capacity to recover the conditional structure as our model does, its estimation accuracy is worse than that of our model, especially for the third eigenfunction.  The fFPCA approach also suffers from prohibitively high computational cost, as we now illustrate.

	\begin{figure}[t]
		\centering
		\includegraphics[width=0.6\textwidth]{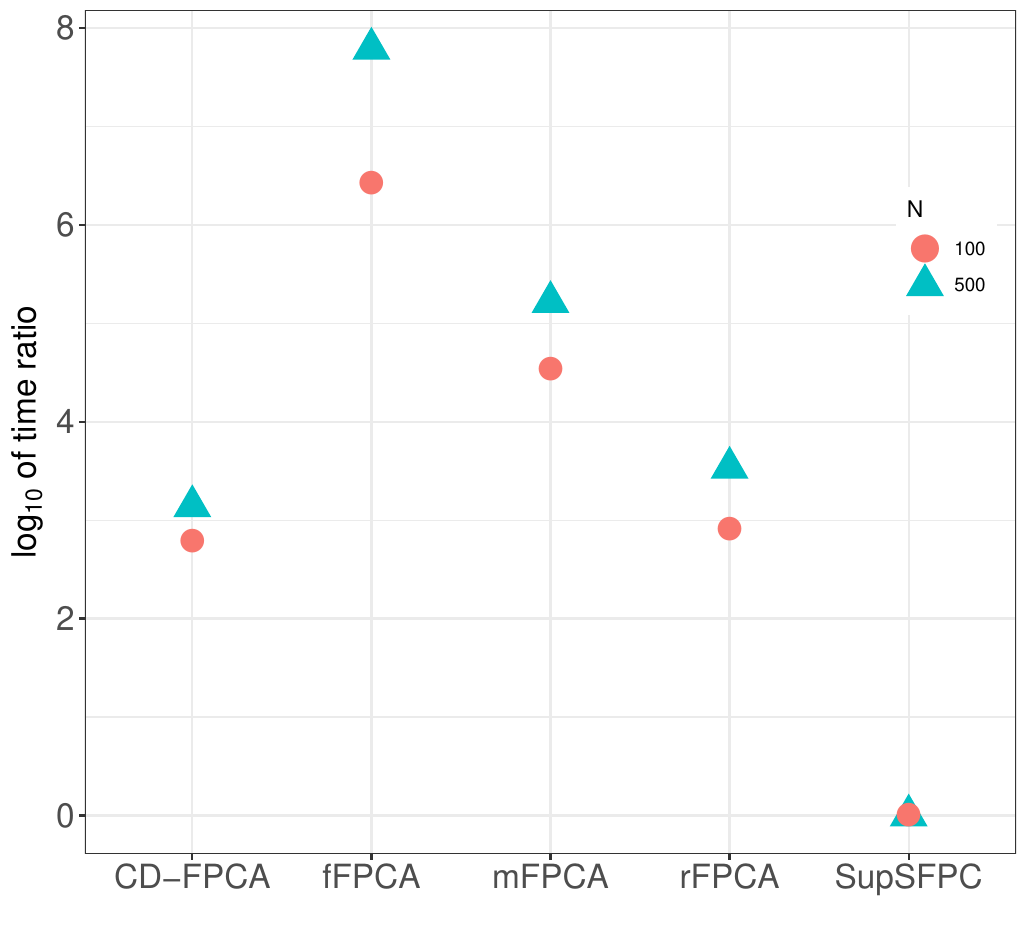}
		\caption{Log relative run time $\log_{10}(T_E/T_{\text{SupSFPC}})$, where $T_E$ denotes the mean run time in seconds for $E\in\{\text{CD-FPCA}, \text{fFPCA}, \text{mFPCA}, \text{rFPCA},\text{SupSFPC}\}$. The round points and triangle points represent different datasets, i.e., $N = 100$ and $N = 500$, separately}
		\label{fig:time_comparison}
	\end{figure}
	Figure~\ref{fig:time_comparison} compares the computational cost of the five methods and  shows the logarithm of computational time ratio $\log_{10}(T_E/T_{\text{SupSFPC}})$, where $T_E$ denotes the mean run time in seconds for $E\in\{\text{CD-FPCA},\text{SupSFPC}, \text{fFPCA}, \text{mFPCA}, \text{rFPCA}\}$. SupSFPC is used as the baseline because it is the fastest method. For reference, SupSFPC took an average of $0.024$ and $0.080$ seconds for the $N=100$ and $N=500$ cases, respectively. Our method is computationally more efficient than all the other methods except SupSFPC. Despite its speed, the SupSFPC approach has limitations, because it performs substantially worse than CD-FPCA (and fFPCA) in terms of MSE. For example, Table~\ref{tab:real_comparison} shows that under SupSFPC  the MSE for  the first eigenfunction is two times higher than under CD-FPCA.  The  local smoother based approaches (fFPCA and mFPCA) are  computationally inefficient. For example, it took about 160 hours for fFPCA to finish one replicate with $N=500$ samples. The high accuracy and comparatively low computational cost of our method means that it is more suitable for application to large datasets than the other approaches. 
	
	
	\subsection{Further Comparision with SupSFPC}

	Section~S.1 of the Supplementary Materials considers further comparisons between CD-FPCA and  SupSFPC,  because  SupSFPC is the only other method considered here which can be applied to large datasets and also incorporates covariate information. We compared the predictive performance of these two models over the simulated dataset in Section~\ref{subsec:simulationdesign}. We also simulated an additional dataset using the SupSFPC model \eqref{eqn:SupSFPC:model}, in which the scores vary linearly with the covariates. Lastly, we simulated a third dataset using the SupSFPC model~\eqref{eqn:SupSFPC:model}, except that the scores vary \textit{quadratically} with the covariates.
	Our conclusion from these simulations is that CD-FPCA performs comparably to SupSFPC when data are generated under the latter model~\eqref{eqn:SupSFPC:model}, and performs substantially better when the linear assumption~\eqref{eqn:SupSFPC:score} is violated. Our findings illustrate the greater flexibility of CD-FPCA compared with SupSFPC, and suggest that the former method should be preferred, except perhaps in the case of very large datasets for which we are also confident that the SupSFPC assumptions are satisfied.

	\section{MODELING ASTRONOMICAL  LIGHTCURVES}\label{sec:application}
	
	\begin{figure}[t]
		\centering
		\includegraphics[width=0.9\textwidth]{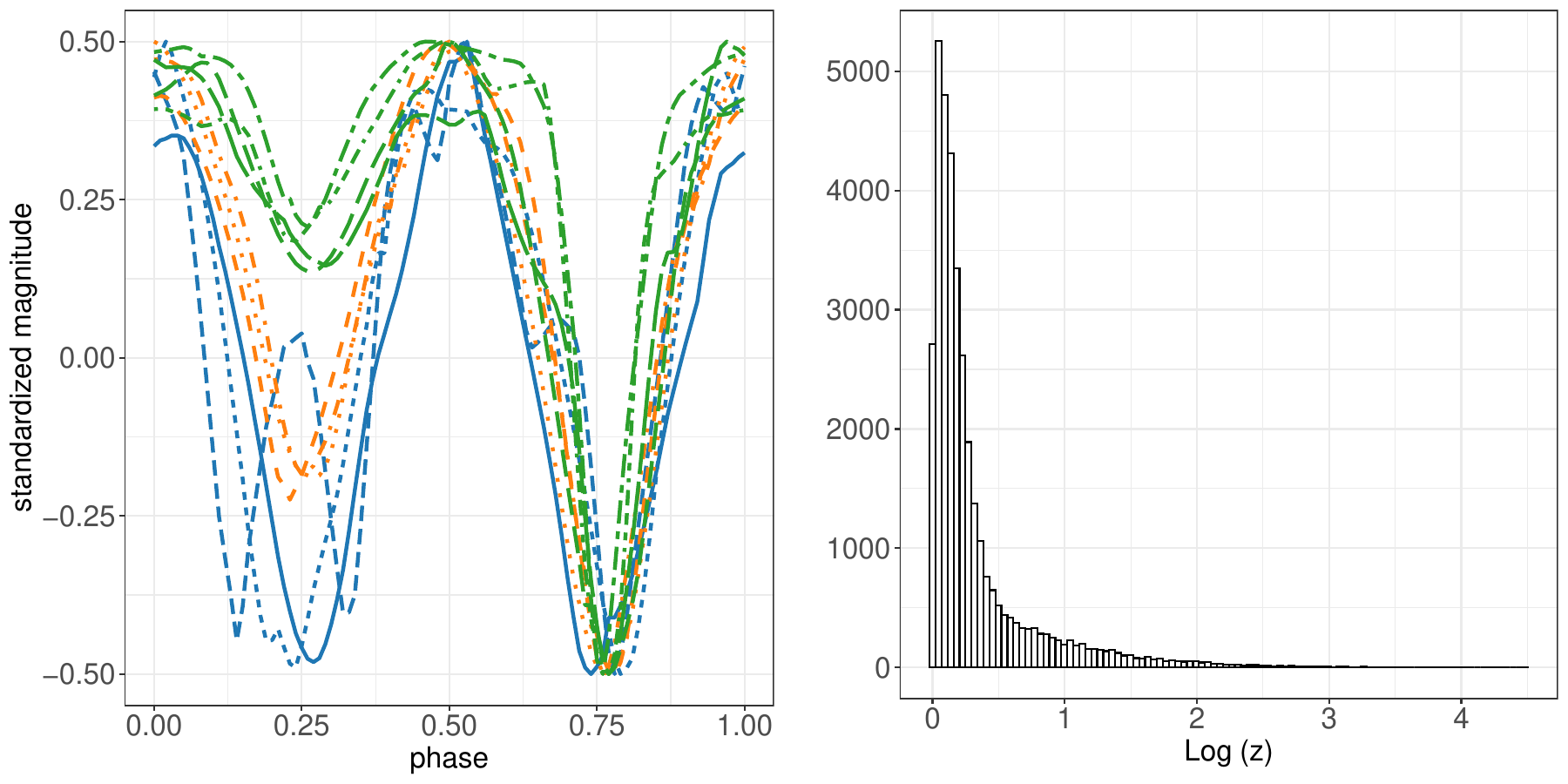} 
		\caption{Left panel: standardized lightcurves from eclipsing binary sources. Lightcurves with covariate values in $[1.00,1.06]$, $[1.34,1.57]$ and $[2.39,88.48]$ are shown in blue, orange, and green, respectively. Right panel: histogram of the logarithm of the covariate $z$.  }\label{fig:realdata:examples}
	\end{figure}
	
	In astronomy, a lightcurve is a time series of the observed brightness of a source, e.g., a star or galaxy. Lightcurves are useful because some astronomical sources vary in brightness over time, and these variations can be used to classify the type of  source or infer its properties, e.g., the period of star pulsations (from which additional physical insights can be gained). One type of variable source is an eclipsing binary system, which is a system of two stars orbiting each other. Many stars visible to the eye are in fact eclipsing binary systems. If the orbits of the two stars lie in the plane that also contains our line of sight then the stars will alternately eclipse each other from our perspective. Binary stars cannot usually be resolved, but the eclipses block some of the light from reaching us and create periodic dips in the observed lightcurve. These characteristics can be used to distinguish eclipsing binary sources from other variable sources and help us to infer properties of the two stars, e.g., their relative masses.

	Our data set consists of $N = 35615$  eclipsing binary lightcurves from the Catalina Real-Time Transient Survey (CRTS) \citep{drake2009first} which were classified by the CRTS team in \citet{drake2014catalina}. Each observed magnitude (brightness) measurement is accompanied by a known measurement error (i.e., standard deviation), that is determined by astronomers based on the properties of the telescope used. The data are publicly available from \url{http://crts.caltech.edu/}. The left panel of Figure~\ref{fig:realdata:examples} shows $10$ standardized eclipsing binary lightcurves from the dataset. The $y$-axis units are standardized magnitude: magnitude is an astronomical measure of the intensity of light from a star, with smaller numbers indicating greater intensity. Standardizing so that all the measurements fall in $[-0.5,0.5]$ is necessary here because we are principally interested in modeling the similar shapes of the lightcurves.  For visual purposes the measurement errors are not plotted.
	The $x$-axes  in the left panel of Figure~\ref{fig:realdata:examples}  are phase of oscillation (as opposed to time),  because eclipsing binary lightcurves are periodic. The period of  oscillation for each lightcurve was found by \citet{drake2014catalina} and is treated as known for the purposes of our analysis. In practice, the periods would have to be estimated, which is itself a challenging inference problem. It is worth noting that the improved modeling we present here could in turn facilitate improved period estimation accuracy in future. 
	
	An important feature in Figure~\ref{fig:realdata:examples} is that for some lightcurves the depth of the two eclipses are similar, and for others  they are very different. This distinction is due to there being different types of eclipsing binary system. Eclipsing binaries are often divided into two classes, contact binaries which are sufficiently close to exchange mass, and detached binaries which are more separated. Contact binaries typically consist of two sources with similar properties (e.g., size and brightness), meaning that the two eclipses are similar. In contrast, detached binaries may have eclipses of any relative size, because the two sources can have completely different properties.
	
	The above considerations raise an important modeling challenge: eclipsing binary lightcurves can be modeled using somewhat similar functions, because they have similar shapes and covariance structures, but it does not make sense to treat them as coming from a completely homogeneous distribution,  as is typically assumed in FPCA. The current solution is to divide eclipsing binaries into contact and detached binaries and treat these groups as homogeneous, but this is still unsatisfactory because the detached binaries group is heterogeneous. Treating eclipsing binaries as homogeneous, means that any models we use to fit them will either be inaccurate or unnecessarily complicated, which in  turn will reduce our ability to classify them, estimate their periods, and learn their other properties. Instead, we use our CD-FPCA method to learn a mean function and a set of covariance matrix eigenfunctions that smoothly vary with the relative depth of the two eclipses. This approach captures the fact that eclipsing binary lightcurves are similar, while also accounting for a physically interpretable difference.

	\begin{figure}[t!]
		\centering
		\includegraphics[width=0.8\textwidth]{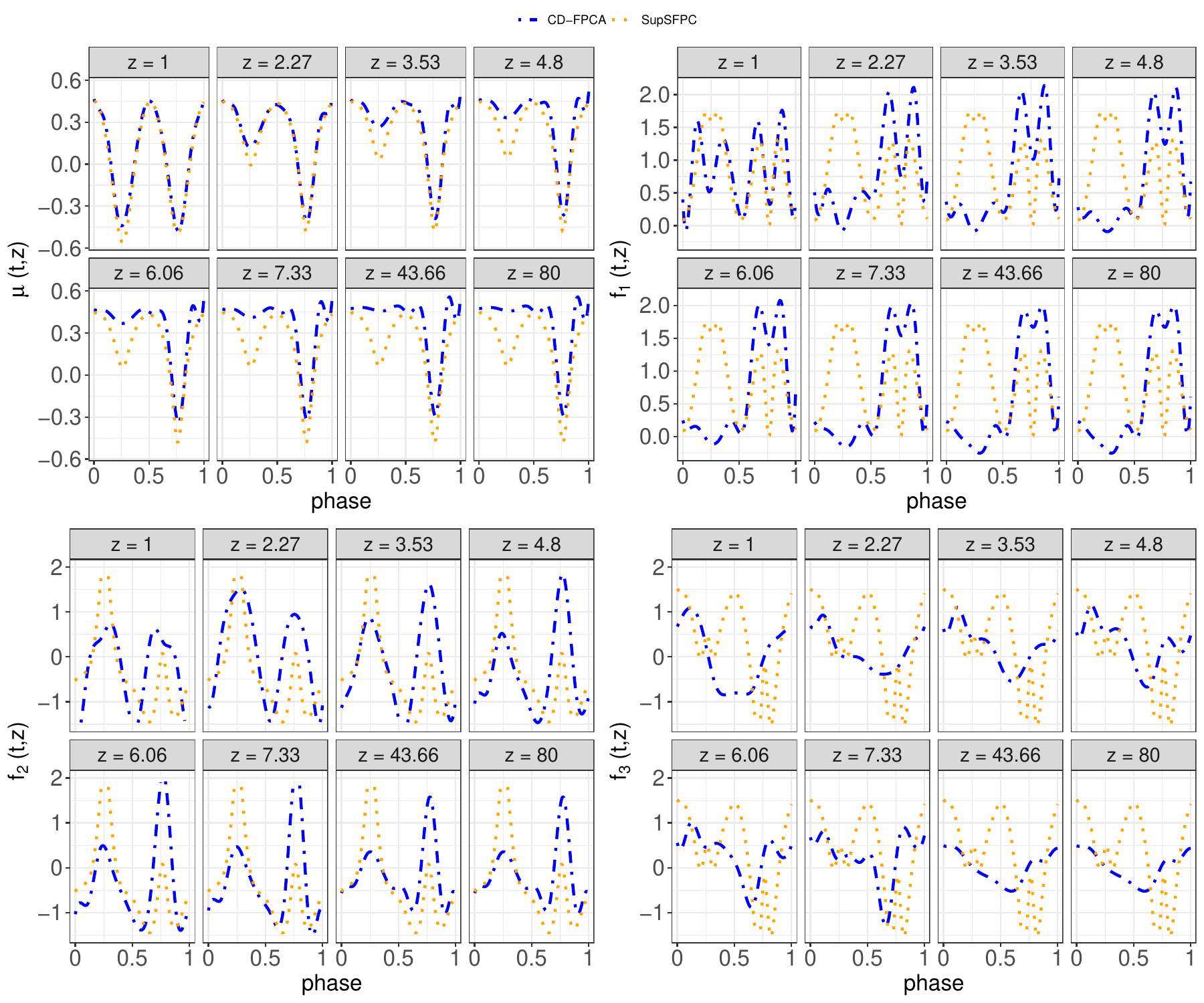} 	
		\caption{(Top left) estimates of the mean function by CD-FPCA (blue dot-dash lines) and SupSFPC (orange dotted lines), for a range of covariate values. (Top right, bottom left, bottom right) estimates of eigenfunctions $f_1$, $f_2$, and $f_3$, respectively, under the CD-FPCA  and SupSFPC methods.}\label{figcomgride_eig}
	\end{figure}

	To implement our approach we need a parameter or covariate $z$ related to the relative depth of the two eclipses of each lightcurve. In practice,  such information may sometimes be available from a separate observation of the eclipsing binary system, e.g., from another  telescope targeting a different light wavelength range. However, in many cases a parameter capturing the relative depth would need to be inferred from the data. For the sake of simplicity, in this work we calculate an approximation of the relative depth of the eclipses of each lightcurve from the data and treat it as a covariate $z$. In particular, we use a simple cubic B-spline approximation to each lightcurve and compute the ratio of the change in standardized magnitude for the larger eclipse and that for the smaller eclipse. For our dataset this gives $z$ values in the range $[1.00,88.48]$, and the histogram of $\log(z)$ is shown in the right panel of Figure~\ref{fig:realdata:examples}. The key point is that this covariate is low dimensional (univariate) but still explains a great deal of the variation between the lightcurves (which have infinite dimension).

	Since SupSFPC only handles regularly spaced data and does not make use of measurement errors, we initially consider  a processed version of the data in which all the observations lie on a regular grid in phase  space. In particular, we use a cubic spline fit to each lightcurve to obtain observations at phases $t_i = \frac{i-1}{m-1}$ for $i=1,\dots,m = 101$, regardless of the number of observations in the original lightcurve. For this gridded data, we do not have measurement errors. 
	The estimated mean function and eigenfunctions under CD-FPCA and SupSFPC are shown in Figure~\ref{figcomgride_eig}.   Compared with the  SupSFPC estimates (orange dotted curves), the estimates by CD-FPCA (blue dot-dashed curves) show improved low-rank representation and detect meaningful features. In particular, in the upper left panel of Figure~\ref{figcomgride_eig}, the estimated mean curves under CD-FPCA have  meaningful interpretations. When $z$ is large, the lightcurves likely correspond to Algol eclipsing systems. As stated in \cite{kallrath2009eclipsing}, the depth of the two minima are very different for Algol lightcurves.  The mean curves estimated under our CD-FPCA approach successfully detect this feature, whereas those under SupSFPC  do not. 
	
	The  eigenfunctions estimated by CD-FPCA are also more interpretable. For different types of the eclipsing system, e.g., contact binaries and detached binaries,  lightcurve variability is  different. Contact binaries typically consist of two similar sources (e.g., in terms of size and brightness), and therefore the patterns seen in the first half phase $[0,0.5)$ and the second half phase $[0.5,1)$ of the lightcurve are usually similar. In contrast, detached binary systems tend to be composed of two sources with differing properties, which means that the variability seen in the first half phase can be quite different to that seen in the second half phase.  The estimated first eigenfunction under CD-FPCA (blue dot-dash lines in the top right panel of Figure~\ref{figcomgride_eig}) captures this feature. When the covariate is small,  the first and second half phase components of the eigenfunctions are similar, which aligns with the interpretation that these lightcurves come from contact binaries. On the other hand, as the covariate increases, the second half phase shows greater relative variability, which is consistent with detached binaries as we now explain further. For Algol eclipsing binary systems (for which we expect $z$ to be large), the smaller star is hotter and bluer and this leads to  stable brightness (relatively smaller fluctuations) in the first half phase. In particular, the smaller dip in the lightcurve during the first half phase may be undetectable partly due to the ``reflection effect'' (a reprocessing of the hotter stars  radiation as it falls on the atmosphere of its companion, increasing the cooler star's luminosity, see \cite{kallrath2009eclipsing}). The estimated first eigenfunction of CD-FPCA  well captures  this ``reflection effect'' becasue they indicate  that for large $z$ the fluctuations in the first half phase are much smaller compared to those in the second half phase. The second eigenfunction has a similar interpretation, and the third captures additional fluctuations.

	
	In Section~S.2 of the Supplementary Materials, we further compare the quality of the estimated  eigenfunctions for lightcurve data by comparing the predictive performance of  CD-FPPCA and SupSFPC.  We find that  CD-FPCA has superior predictive performance, which again indicates that our model is more appropriate for this application. We anticipate that our approach will similarly offer better performance when inferring periods and other physical properties of binary systems, especially when observations are sparse and noisy. In Section~S.2, we apply our method to the raw (non-gridded) data (for which SupSFPC is not applicable), and it demonstrates similar performance.

	\section{DISCUSSION}\label{sec:discussion}
	
	Our CD-FPCA method offers an attractive option for analyzing large functional datasets in which both the mean and covariance function vary smoothly with covariates. It can flexibly incorporate this type of dependence and has substantially lower computational cost than  popular local smoother based covariate adjustment approaches, e.g., \cite{jiang2010covariate}, \cite{jiang2011functional}, \cite{zhang2013time}, and \cite{zhang2016sparse}. While the SupSPFC approach proposed by \citet{li2016supervised} has even lower computational cost than CD-FPCA, it does not achieve the same levels of flexibility and accuracy. Indeed, in both our  simulation study and our real data analysis, CD-FPCA performed better than SupSFPC in all the accuracy comparisons that we considered, e.g., mean and eigenfunction estimation, and prediction accuracy (except when the restrictive assumptions of SupSFPC were exactly satisfied, in which case the two methods were comparable, see the first and third rows of Table~S.2). Furthermore, CD-FPCA can handle irregular and unbalanced observation times  and incorporate measurement errors, whereas SupSFPC can not. 
	
	In this work, we have focused on illustrations with univariate covariates,  but our model formulation in Section~\ref{sec:model} is general. The proposed method can be  extended to multivariate covariate case via using a basis generated by a reproducing kernel \citep{rasmussen2006gaussian}.	A reproducing kernel  $k(\vz, \vz')$ has the reproducing property $\langle k(\cdot, \vz) , k(\cdot, \vz')  \rangle = k(\vz, \vz').$ For example, $k(\vz, \vz') = \exp(-\delta\|\vz - \vz'\|_2^2)$ is a choice of reproducing kernel.  A reproducing kernel induces a reproducing kernel Hilbert space (RKHS), and the RKHS-norm of a function measures the roughness of the function.	By selecting a set of inducing points	$\tilde{\vz}_1,\cdots, \tilde{\vz}_q$ in the range of $\vz$, the elements of $\mC(\vz;\vbeta)$ can be approximated by a basis expansion $C_{ij}(\vz) = \sum_{k=1}^{q} \beta_{ijk}k(\vz, \tilde{\vz}_j)$. The RKHS-norm of $C_{ij}(\vz)$ can be used for regularization. Details of this extension is left for future research. The proposed method can also be easily extended to  categorical covariates. Suppose the covariate $\vz$ takes $K$ values $1, \dots, K$. Then, we use the basis representation  $\mC(\vz;\vbeta) = \sum_{k=1}^{K} \mC_k I(\vz = k)$, where $I(\cdot)$ is the indicator function, and the matrix $\mC_k = \big(\beta_{kij}\big)\in\mathbb{R}^{w\times r}$ has  scalars $\beta_{kij}$ as its entries. In this case, there is no need to introduce roughness penalty over $\vz$, and we need only to set $\lambda_{\vz}^{(\mu)}= \lambda_{\vz}=0$ in our algorithm.
	
	Regarding further methodological development, one could construct a map from Euclidean space to the Stiefel manifold, as opposed to the symmetric positive semi-definite rank $r$ matrices   (see \eqref{EC1}), because the underlying FPC coefficients matrix, denoted $\mTheta_{\vz}$,  lies on the Stiefel manifold (where, more specifically, $\mTheta_{\vz}$ is such that $\vf(t,\vz)=\mTheta_{\vz}^T \vb(t)$).  Indeed, the Stiefel manifold is  a more frequently used structure than the symmetric positive semi-definite rank $r$ matrices, and optimization methods on the Stiefel manifold  have been well developed, e.g., \cite{boothby1986introduction}, \cite{balogh2004some}, \cite{nishimori2005learning}, \cite{wen2013feasible}. There may be advantages to a Stiefel manifold based approach compared with our method proposed here, or vice versa, but this needs to be further investigated.

	\section*{Acknowledgment}
	The authors thank the editor, the associate editor, and three anonymous reviewers for constructive comments that helped significantly improve this work.  Ding's research was accomplished during his visit to Department of Statistics, Texas A\&M University. Ding would like to thank Professor Joe Newton for his financial support during part of his visit.
	He's research was partially supported by National Natural Science Foundation of China (No.11801561).
	
	\bibliographystyle{apalike}
	
	\bibliography{Bibliography-MM-MC}
	
	\appendix

\end{document}